\documentclass[doublecol,figures]{epl2}

\usepackage{amsmath}
\usepackage{amsfonts}
\usepackage{amssymb}
\usepackage{mathrsfs}
\usepackage{latexsym}
\usepackage{graphicx}
\usepackage{graphics,psfrag}
\usepackage{graphicx,psfrag}

\newtheorem{theorem}{Theorem}

\newenvironment{proof}[1][Proof]{\noindent\textbf{#1.} }{\ \rule{0.5em}{0.5em}}

\title{Nonviolation of Bell's Inequality in Translation Invariant Systems}

\author{Thiago R. de Oliveira\thanks{E-mail: \email{tro@if.uff.br}} \and A. Saguia \and M. S. Sarandy}
\shortauthor{Thiago R. de Oliveira, A. Saguia, M. S. Sarandy}

\institute{Instituto de F\'{\i}sica, Universidade Federal Fluminense, Av. Gal. Milton
Tavares de Souza s/n, Gragoat\'a 24210-346, Niter\'oi, RJ, Brazil}


\abstract{
The nature of quantum correlations in strongly correlated systems has been a subject of 
intense research. In particular, it has been realized that entanglement and quantum discord 
are present at quantum phase transitions and are able to characterize them. Surprisingly, 
it has been shown for 
a number of different systems that qubit pairwise states, even when highly entangled, do 
not violate Bell's inequalities, being in this sense local. Here we show that such a local 
character of quantum correlations traces back to the the monogamy trade-off obeyed 
by bipartite Bell correlations, being in fact general for translation invariant systems. We 
illustrate this result in a quantum spin chain with a soft breaking of translation symmetry. 
In addition, we provide an extension of the monogamy inequality to the $N$-qubit scenario, 
showing that the bound increases with $N$ and providing examples of its saturation through 
uniformly generated random pure states.
}

\begin{document}

\maketitle

\section{Introduction}
Correlations are a central concept in science, if not the essence of it. They typically arise 
from interactions, being responsible for plenty of phenomena. The antiferromagnet exchange 
interaction, for example, impose a correlation between the poles of two magnets: they align 
in opposite direction. Indeed, many of the most interesting effects in condensed matter has 
their origins in strongly correlated systems. A standard representative is continuum phase 
transitions, where a macroscopic drastic change in a system occurs due to the onset of 
long-range correlations: a finite magnetization as the temperature decreases, for example. 
On the other side, in the last decade, the study of correlations {\it per se} has gained a 
great deal of attention due to the applications of quantum computation and information. The 
interest comes from the fact that quantum mechanics allows for more general correlations than 
those available in classical systems. Such correlations, which can be quantified, e.g., by 
entanglement or quantum discord, are general resources for protocols to implement quantum 
tasks~\cite{Nielsen:book,Dakic:12,Gu:12,Madhok:12}. 

Recently, condensed matter and quantum information communities have exchanged knowledge about 
correlations~\cite{Amico08,Modi:12}. In particular, the roles played by both entanglement and 
quantum discord at a quantum phase transition (QPT) have generated great 
interest~\cite{Osterloh02,Wu:04,Sarandy:09}. This motivation is based on the fact that QPTs 
occur at zero temperature, where the state of the system is typically pure, and then should be 
driven by quantum correlations. Note that long-range correlation (spin-spin, for example) is at
the origin of the universality paradigm in critical phenomena.  Remarkably, while classical 
correlation and discord are 
long-ranged~\cite{Maziero:12}, bipartite entanglement is not.
Such an important observation can be understood in terms of a monogamy trade-off: 
bipartite entanglement obeys a much stricter distribution law than classical correlation and 
quantum discord~\cite{Streltsov:12,Braga:12}. In turn, monogamy inequalities may imply in 
constraints of remarkable consequences for the behavior of correlations. 

However, besides entanglement and discord, other nonclassical correlations are allowed by 
quantum mechanics, e.g., Bell correlations. The violation of a Bell inequality indicates 
nonlocality. This is a concept independent of other kind of correlations. It is clearly 
inequivalent to quantum discord, since Bell inequalities are satisfied by separable states, 
while discord may be nonvanishing. Moreover, it is not equivalent to entanglement either, 
since the the presence of entanglement for mixed states does not necessarily imply in a 
violation of Bell's inequalities~\cite{Werner89}. A natural challenge is therefore to 
understand the nature of each correlation measure in strongly correlated quantum systems. 
Indeed, different correlations may be associated with resources for distinct tasks. 
As for entanglement and discord, characterization of QPTs through Bell correlations between 
qubits through the Clauser-Horne-Shimony-Holt (CHSH) inequality has recently been 
addressed~\cite{Altintas12,Batle10,Deng11,Justino11}. However, it is rather surprising that 
the CHSH inequality, when applied for any pair of qubits in critical spin latices, has been 
observed to be nonviolated in a number of different systems (see, e.g., 
Refs.~\cite{Altintas12,Batle10,Deng11,Justino11,Wang:02}), even for highly entangled spin pairs. 
Thus, it appears that, for a two-spin system within a lattice in the thermodynamic limit, 
the long-range correlations typical of a QPT, which ensures a considerable amount of pairwise
entanglement, is unable to reveal nonlocal effects. It has been an open problem to identify the 
origin of this general absence of violation. 

In this work, we give a simple explanation for this negative result. More specifically, we 
establish a {\it no-go} theorem showing that {\it the monogamy of Bell correlations imposed by 
the CHSH inequality forbids the manifestation of nonlocality by any entangled spin pair in 
a translational invariant lattice}. Therefore, a basic principle, namely, monogamy, enforces 
the nature not only of entanglement and quantum discord, but also introduces a new kind of 
general restriction in the behavior of Bell correlations. As an illustration, we consider a 
critical spin chain with a soft breaking of translation symmetry. Moreover, {\it we also extend 
the monogamy of CHSH inequality for $N$ particles and investigate its behavior 
for random states}. 

\section{Monogamy of entanglement and nonlocality} 
It is known that entanglement cannot be freely shared among the parts of a composite 
system, i.e., it is monogamous. Indeed, in a a pure state of 
three qubits A, B and C, if A is entangled with B, it cannot be 
entangled with C. For mixed states, monogamy is not strict, but there are 
bounds on its validity. The most famous relation expressing entanglement monogamy 
has been obtained by Coffman {\it et al.}~\cite{Coffman00} for 
three particles and the generalized by Osborne and Verstraete~\cite{Osborne06} 
for $N$ qubits, reading 
$C^{2}_{1,2}+\cdots +C^{2}_{1,N}\leq C^{2}_{1,2\cdots N}$, 
with $C_{i,j}$ denoting entanglement between qubits $i$ and $j$ as measured 
by concurrence~\cite{Wootters:98}, with the upper bound $C_{1,2\cdots N}$ denoting the 
concurrence between qubit 1 and all the rest of the system. Note that this upper bound 
does not increase with the number of particles, since $C^{2}_{1,2\cdots N}\leq 1$ for
any $N$. For the case of quantum discord, a monogamy relationship such as that obeyed by concurrence 
does not hold in general~\cite{Streltsov:12}, but a softer monogamy relationship has been recently 
established~\cite{Braga:12}. 

Concerning Bell correlations, let us consider three spins-1/2 particles in a quantum state described by a 
density operator $\rho$ and a standard Bell experiment where each party chooses two directions in a 
Stern-Gerlach apparatus. We then define the Bell operator $B_{ij}$ acting on the Hilbert space 
${\cal{H}}={\cal{H}}_i \otimes {\cal{H}}_j$ $(i,j=1,2,3)$ for parties $i$ and $j$ as
\begin{equation}
B_{ij}=\hat{a}_i \cdot \vec{\sigma} \otimes (\hat{a}_j+\hat{a}_j^\prime) \cdot \vec{\sigma} + 
\hat{a}_i^\prime \cdot \vec{\sigma} \otimes (\hat{a}_j-\hat{a}_j^\prime) \cdot \vec{\sigma} \, ,
\end{equation}
where $\hat{a}_i,\hat{a}_i^\prime,\hat{a}_j,\hat{a}_j^\prime$ are unit vectors in $\mathbb{R}^{3}$ and 
$\vec{\sigma}$ denotes a Pauli vector operator. By taking the expectation value 
$\langle B_{i j} \rangle=\textrm{Tr}\left(\rho \,  B_{i j} \right)$ and maximizing over measurements 
directions in 
$\langle B_{i j} \rangle$, Toner and Verstraete~\cite{Toner06} have 
established the monogamy trade-off \footnote{Although the proof suppose a pure state of three qubits, 
the result also holds for mixed states by a convexity argument.}
\begin{equation}
\max_{1,2}\langle B_{12}\rangle^{2}+\max_{1,3}\langle B_{13}\rangle^{2}\leq8,
\label{eq: Bell Monog}
\end{equation}
where $\max_{i,j} = \max_{(\hat{a}_i,\hat{a}_i^\prime),(\hat{a}_j,\hat{a}_j^\prime)}$. 
Then, it follows from Eq.~(\ref{eq: Bell Monog}) that, if parties $1$  
and $2$ violate the CHSH inequality, namely, $\left| \langle B_{12} \rangle \right| > 2$, 
then parties $1$ and $3$ must necessarily obey it, i.e., $\left| \langle B_{13} \rangle \right| \le 2$. 
For the case of two qubits, an analytic expression for the maximum violation of the Bell 
inequality was obtained in Ref.~\cite{Horodecki95}. By defining the Bell correlation as
\begin{equation}
{\cal B}_{ij}=\max_{i,j}|\langle B_{i,j}\rangle| , 
\label{eq: Bell Operator}
\end{equation}
it is found that 
\begin{equation}
{\cal B}_{ij}=2\sqrt{u+u'} \, ,
\label{Analytical}
\end{equation}
where $u$ and $u'$ are the two largest eigenvalues of the matrix $U=T^{T}T$, with 
$T$ labeling the $3\times3$ matrix built from the correlations $t^{uv}=\textrm{Tr}\left(\rho \,\sigma_u \otimes \sigma_v\right)$. 
The monogamy bound (\ref{eq: Bell Monog}) can be directly derived from Eq.~(\ref{Analytical}). 
In such a proof, the directions of $1$ that maximize the violations with $2$ and $3$ can be different (this 
is not allowed {\it a priori} by a similar proof in Ref.~\cite{Kurzynski11}). Moreover, the possibility of 
particle $1$ to share two singlets, one with particle $2$ and the other with particle $3$, which would violate the 
bound, is naturally forbidden, since each part can only have one qubit.

\section{QPTs, correlations, and CHSH inequality}
The investigation of entanglement and 
quantum discord at QPTs generated a large amount of research (see, e.g., Refs.~\cite{Amico08,Modi:12} 
and references therein), which has recognized them as useful measures to characterize a QPT. 
In this context, some remarkable properties surprisingly appeared. 
For example, it has been observed that bipartite entanglement
between individual particles is short-ranged. As mentioned, this result can be understood as 
consequence of monogamy: for one particle to have a reasonable amount of entanglement with a neighbor, 
its entanglement with particles far away must be negligible. On the other hand, genuine multipartite 
entanglement is typically present at quantum criticality (see, e.g., Ref.~\cite{Chandran07}). 
Concerning pairwise quantum discord, its long-range behavior in critical systems~\cite{Maziero:12}  
can also be traced back to the softer monogamy obeyed by quantum discord~\cite{Braga:12}.

Very recently, the CHSH inequality has been used as a tool 
to characterize QPTs~\cite{Batle10,Altintas12,Justino11,Deng11}. In turn, 
it has successfully performed this task, indicating QPTs through 
nonanalyticities in the derivatives of the Bell correlations. 
We observe that this method works whether or not violations of the CHSH inequality actually occur. 
In fact, for all of these works, it was found that, given any two qubits, 
violation of the CHSH inequality is never achieved, even for highly entangled qubit 
pairs. Remarkably, this curious behavior is also rooted in a monogamy relation, which is 
formalized in Theorem~\ref{t1} below.

\begin{theorem}
{\it Consider an arbitrary $N$-qubit composite system arranged in a lattice with translation 
invariance. In such a system, any (pure or mixed) state for two qubits cannot violate the CHSH 
inequality.} 
\label{t1}
\end{theorem}

\begin{proof}
Let us denote the two-qubit reduced density operator for qubits $i$ and $j$ as $\rho_{i,j}$. 
Then, translation invariance implies that $\rho_{i,j}=\rho_{k,l}$ if $i-j=k-l$, i.e. the 
density operator depends only on the distance $r$ between the qubits. Now, take a subsystem 
composed by any two qubits at an arbitrary distance $r$, labeling them as $1$ and $2$. Then, 
join to this subsystem a third qubit (labeled as $3$) also at distance $r$ from qubit $1$, which 
can be represented by

\begin{figure}[th]
\centering {\includegraphics[angle=0,scale=0.3]{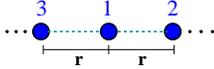}}
\caption{(Color online) Tripartite subsystem with qubit $1$ at fixed distance $r$ from qubits $2$ and $3$. }
\label{f1}
\end{figure} 

For such a state, we have in particular that $\rho_{1,2}=\rho_{1,3}$, which implies that 
$\langle B_{1,2}\rangle = \langle B_{1,3}\rangle$ for any fixed measurement direction. 
Therefore, from Eq.~(\ref{eq: Bell Monog}), we obtain that
\begin{equation}
2 \max_{1,2}\langle B_{1,2}\rangle^{2} \leq 8 \,\,\, \Longrightarrow \,\,\, 
\max_{1,2}\langle B_{1,2}\rangle \le 2.
\label{eq: Bell Monog-2}
\end{equation}
Hence, translation invariance and monogamy of Bell correlations together prevent the violation of 
the CHSH inequality for any two qubits of the system. 
\end{proof}

We should also mention previous works on the relation between symmetries and violation
of Bell inequalities. In Ref.~\cite{Terhal03}, it is established a relation between the possibility of 
sharing a state and the violation of a Bell inequality. As a consequence, no two-qubit pair can violate
the CHSH inequality in a system of $N$ qubits with permutation symmetry. Using this result, it has been 
shown in Ref.~\cite{Ramanathan11} that if one is only able to do collective measurements in $N_X$ spins than one
need at least $S_X > N_X$ measurement directions on each part to violate a Bell inequality.
It is worth mentioning that Werner himself, already in 1989, using a more abstract formalism in
terms of $C^\star$ algebras, considered the relation between symmetries and violation of Bell inequalities~\cite{Werner89B}.

Theorem 1 can also be rederived from the results of Ref.~\cite{Terhal03}. However, we have obtained 
it here from a new (and much simpler) approach, which is based on the monogamy of Bell 
inequalities. This not only simplified the proof, but also provided a connection between symmetry and 
monogamy of Bell correlations, which is potentially suitable for applications in condensed matter systems, 
as will be shown in the next section. 

\section{Bell correlations in spin chains with translation symmetry breaking}
Let us illustrate the results of the previous section in the dimerized Heisenberg spin 
chain, which is characterized by Heisenberg spin interactions where a soft breaking 
of the translation invariance is induced by the split of the system into two 
sublattices with spin interactions of strengths $J_1$ and $J_2$, respectively. 
Then, the dynamics is governed by the 
Hamiltonian
\begin{equation}
H = \sum_{i=1}^{N/2} \left( J_1 \, \vec{\sigma}_{2i-1} \cdot \vec{\sigma}_{2i} 
+ J_2 \, \vec{\sigma}_{2i} \cdot \vec{\sigma}_{2i+1} \right),
\label{HJ1J2}
\end{equation}
where $\vec{\sigma} = (\sigma_x,\sigma_y,\sigma_z)$ denotes the Pauli operator 
vector and $N$ is the number of spins in the chain (taken as an even number). Moreover, 
it is assumed $J_1>0$ and periodic boundary conditions are adopted, i.e. $\vec{\sigma}_{N+1} =  \vec{\sigma}_{1}$. 
At temperature $T=0$, the chain exhibits the following magnetic behavior: 
(i) in the limit $(J_2/J_1) \rightarrow 0$, the ground state is just an ensemble
of $N/2$ uncoupled dimers around strong bonds, with an energy gap separating the 
ground state from the first excited state; (ii) at the isotropic Heisenberg 
point $J_2/J_1 = 1$, the energy gap closes, with the system in a quantum critical regime; 
(iii) for the case of $J_2/J_1 > 1$, the system is noncritical, achieving another   
strongly dimerized ground state in the limit $(J_2/J_1) \rightarrow \infty$;  
(iv) for the case of $J_2/J_1 < 0$, the dimerization also occurs, with dimers coupled by ferromagnetic 
bonds. Entanglement for the dimerized chain has been discussed in Refs.~\cite{Chen:06,Hao:06}.

In order to evaluate the Bell correlations and then to investigate the monogamy 
bound~(\ref{eq: Bell Monog}), it is convenient to rewrite the eigenvalues $u$ and $u'$ defined in 
Eq.~(\ref{Analytical}) in terms of correlation functions. In general, this may lead to an expression for 
${\cal B}_{ij}$ that is typically  cumbersome. However, the symmetry of the Heisenberg 
interaction strongly constrains ${\cal B}_{ij}$. In fact, the matrix $T$ is already 
diagonal, with $t^{xx}$ , $t^{yy}$ and $t^{zz}$ as its diagonal elements and 
$t^{yy}=t^{xx}=t^{zz}$. Therefore, we can write ${\cal B}_{ij}$ as 
\begin{equation}
{\cal B}_{ij}=2 \sqrt{2} \left| \langle\sigma_{i}^{z}\sigma_{j}^{z}\rangle \right| .
\label{Bell-J1-J2}
\end{equation}

\begin{figure}[!ht]
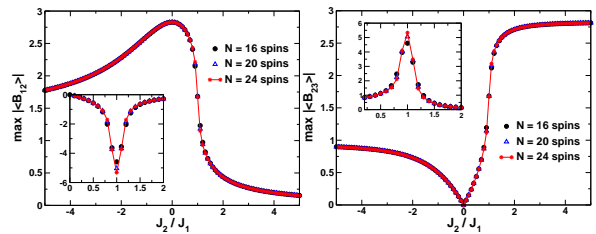

\begin{centering}
\psfrag{$\ell$}{$\ell$}
\includegraphics[scale=0.16]{g2.eps}\includegraphics[scale=0.16]{g3.eps}
\par\end{centering}
\caption{(Color online) Bell correlations ${\cal{B}}_{12}$ and ${\cal{B}}_{23}$ as a function of $J_2/J_1$. Insets: First derivatives of 
${\cal{B}}_{12}$ and ${\cal{B}}_{23}$ with respect to $J_2/J_1$.} 
\label{f2f3}
\end{figure}

We then compute the Bell correlations ${\cal B}_{ij}$ for finite chains up to $N=24$ sites via exact diagonalization 
(by employing the power method). Monogamy of Bell correlations can be expressed here by the fact that ${\cal{B}}_{12}$ and ${\cal{B}}_{23}$, which 
act on distinct sublattices, cannot simultaneously violate the CHSH inequality. This is shown in Fig.~\ref{f2f3}, 
where it is plotted ${\cal{B}}_{12}$ and ${\cal{B}}_{23}$ as a function of the ratio $J_2/J_1$. Note that convergence to 
the thermodynamical limit is fastly achieved, with the plots for $N=16,\,20,$ and $24$ spins approximately yielding the 
same curve. In particular, for the Heisenberg point $J_1=J_2$, we have a translation invariant system, which implies in   
nonviolation of the CHSH inequality for both ${\cal{B}}_{12}$ and ${\cal{B}}_{23}$, as ensured by Theorem~\ref{t1}. 
Moreover, observe that the first derivatives 
of ${\cal{B}}_{12}$ and ${\cal{B}}_{23}$ with respect to $J_2/J_1$, which are plotted in the insets of Fig.~\ref{f2f3}, are 
able to detect the QPT at $J_1=J_2$ through a nonanalyticity at the critical point. Concerning the monogamy bound, it 
is made explicit in Fig.~\ref{f4}, with the sum of squared 
Bell correlations $B_s={\cal{B}}_{12}^2+{\cal{B}}_{23}^2$ always being limited by $8$. 
Moreover, the saturation of the monogamy inequality is obtained in the limits 
of strong dimerization, given by $(J_2/J_1) \rightarrow 0$ and $(J_2/J_1) \rightarrow \infty$. In such dimer limits, either 
${\cal{B}}_{12}$ or ${\cal{B}}_{23}$ saturates the Tsirelson's bound $2\sqrt{2}$, with the other Bell correlation vanishing. 

\begin{figure}[!ht]
\begin{centering}
\psfrag{$\ell$}{$\ell$}
\includegraphics[scale=0.26]{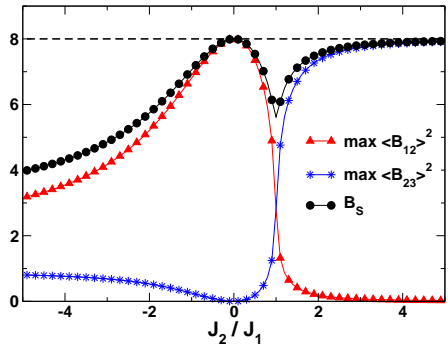}
\par\end{centering}
\caption{(Color online) Sum of squared Bell correlations (denoted as $B_s={\cal{B}}_{12}^2+{\cal{B}}_{23}^2$) as a function of $J_2/J_1$, 
in comparison with ${\cal{B}}_{12}^2$ and ${\cal{B}}_{23}^2$.} 
\label{f4}
\end{figure}

Is there a monogamy bound for N particles, which may restrict
further the maximum value of the bell inequality? If such monogamy
bound is similar to the entanglement monogamy, where the upper bound
does not increase with the number of particles $N$, then the value of
the Bell correlation would have to decrease as the number of particles
increase. However that is not the case, as we will show next.

\section{Monogamy of Bell's inequality for $N$ parties}
It has been shown in Ref.~\cite{Toner06} that 
the monogamy bound for three particles can be obtained from two lemmas: (i) For any Bell inequality 
in the setting where $N$ parties each choose from two dichotomic measurements, we have that its maximum 
value is achieved by a state that has support on a qubit at each site. Such a state can be assumed to have 
real coefficients and the observables chosen are real and traceless; (ii) For any pure state 
$|\psi_{123}\rangle\in\mathbb{C}^{2}\otimes\mathbb{C}^{2}\otimes\mathbb{C}^{2}$ with real probability 
amplitudes, we have that
$\max_{1,2}\langle B_{1,2}\rangle=2\sqrt{1+\langle\sigma_{1}^{y}\sigma_{2}^{y}\rangle^{2}-\langle\sigma_{1}^{y}\sigma_{3}^{y}\rangle^{2}-\langle\sigma_{2}^{y}\sigma_{3}^{y}\rangle^{2}}$.
Bearing in mind these two lemmas, we can derive inequality~(\ref{eq: Bell Monog}) by simple algebra. Indeed, one first obtain 
$\max_{1,2}\langle B_{1,2}\rangle^{2}+\max_{1,3}\langle B_{1,3}\rangle^{2}=8(1-\langle\sigma_{2}^{y}\sigma_{3}^{y}\rangle^{2})$, 
since $\langle\sigma_{1}^{y}\sigma_{2}^{y}\rangle^{2}-\langle\sigma_{1}^{y}\sigma_{3}^{y}\rangle^{2}$
will cancel out with $\langle\sigma_{1}^{y}\sigma_{3}^{y}\rangle^{2}-\langle\sigma_{1}^{y}\sigma_{2}^{y}\rangle^{2}$, which is 
present in $\max_{1,3}\langle B_{1,3}\rangle^{2}$. Then, by neglecting $\langle\sigma_{2}^{y}\sigma_{3}^{y}\rangle^{2}$, 
we obtain the upper bound in inequality~(\ref{eq: Bell Monog}). We will now show that this bound can be
easily extended to $N$ particles, as provided by Theorem~\ref{t2} below. 

\begin{theorem}
{\it Consider an arbitrary $N$-qubit composite system described by a density operator $\rho$, with $N>2$.  Then, Bell correlations computed 
with respect to $\rho$ obey the monogamy inequality} 
\begin{equation}
\sum_{M=2}^{N}\max_{1,M}\,\langle B_{1,M}\rangle^{2}\leq4(N-1).
\label{CHSH-N}
\end{equation}
\label{t2}
\end{theorem}
\begin{proof}
In a
system of $N$ particles we have two possibilities: (i) particle $1$ does
not violate the CHSH inequality with any other party, so each of the terms in the sum
obeys $\max_{1,M}\langle B_{1,M}\rangle^{2}\leq4$, implying in inequality~(\ref{CHSH-N}); 
(ii) particle $1$ violates the CHSH inequality with another single party,
say particle 2 [violation with more than a single party is forbidden by inequality~(\ref{eq: Bell Monog})]. 
Then, by adding inequality~(\ref{eq: Bell Monog}) (for particles 1, 2, and 3) with 
$\max_{1,M}\langle B_{1,M}\rangle^2\leq 4$ (for all $M \ne 2,\,3$), 
we obtain inequality~(\ref{CHSH-N}).
\end{proof}

Differently from entanglement, the bound provided by~(\ref{CHSH-N}) increases with $N$. In fact, it allows
for long-range Bell correlations saturating the classical limit ($\mathcal{B}_{ij}$=2). A simple example of 
saturation is a ferromagnetic product state, with all spins either up or down, which implies in $\mathcal{B}_{ij}$=2 
$(\forall i,j)$. Besides, we can find applications for the monogamy bound~(\ref{CHSH-N}) in more general states, 
e.g., uniformly generated random pure states. For random states up to $N=6$ qubits, we can find out examples very 
close to saturation, with the sum of squared Bell correlations achieving the upper limit $4(N-1)$. However, these 
states appear to be very rare in large Hilbert spaces, as shown in the histogram plotted in Fig.~\ref{f5}. For $N=4$, 
for example, we find the order of $0.01\%$ for states saturating the bound in $90\%$ or more. For larger $N$, saturation 
is even sparser since, while the bound increases with $N$, both the mean value and width of the distribution decrease with $N$ 
(see inset). The computation of the histogram has been performed by randomly choosing $5 \times 10^7$ states in Hilbert space. For the 
system sizes considered in Fig.~\ref{f5}, this number of states already provides an excellent convergence for the mean value 
$(B_s^2)_m$ of the distribution, as displayed in Table~\ref{table1}.

The bound provided in Eq.~(\ref{CHSH-N}) may also be derived from other techniques, such as that used in Ref.~\cite{Kurzynski11}, 
even though it has {\it not} explicitly been obtained there. Besides being lengthier for this specific purpose, 
Ref.~\cite{Kurzynski11} requires that the directions chosen for particle 1 would have to be the same, for all pairs of particles. 
Here, this restriction is not necessary. Moreover, we have also shown the tightness of the bound and applied it for random states.

\begin{figure}[!ht]
\begin{centering}
\psfrag{$\ell$}{$\ell$}
\includegraphics[scale=0.22]{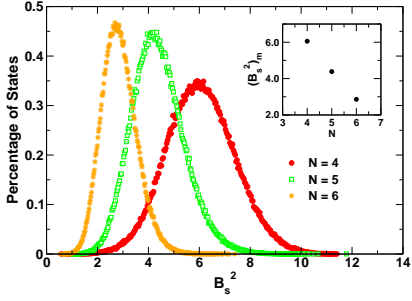}
\par\end{centering}
\caption{(Color online) Sum $B_s^2 = \sum_{M=2}^{N}\max_{1,M}\,\langle B_{1,M}\rangle^{2}$ of squared Bell correlations in random 
states of $N$ spins-1/2 particles. Inset: Mean value of the distribution as a function of the number of spins.} 
\label{f5}
\end{figure}
 
\begin{table}[hbt]
\caption[table1]{Mean value $(B_s^2)_m$ of the distribution as a function of the number of qubits for random pure states. }
\vspace{0.2cm}
\begin{tabular}{ccccc}
\hline
${\textrm{States}}$ &     $N=3$    &      $N=4$     &      $N=5$   &     $N=6$      \\ \hline
    100                & \,7.03812\, & \,5.89693\,  & \,4.32549\, & \,2.81190\,   \\ 
    1 000              & \,6.88858\, & \,5.94993\,  & \,4.38376\, & \,2.83165\,   \\ 
    100 000            & \,6.93172\, & \,6.05699\,  & \,4.38627\, & \,2.86320\,   \\ 
    1 000 000          & \,6.93124\, & \,6.05616\,  & \,4.38381\, & \,2.86035\,   \\ 
    25 000 000         & \,6.93314\, & \,6.05559\,  & \,4.38435\, & \,2.86022\,   \\ 
    50 000 000         & \,6.93337\, & \,6.05556\,  & \,4.38435\, & \,2.86022\,   \\ \hline
\end{tabular}
\label{table1}
\end{table}

\section{Conclusion} 
In conclusion, we have shown that not only the behavior of entanglement and quantum discord at QPTs are enforced by a monogamy trade-off, but also of nonlocality as given in terms of Bell correlations. In particular, the monogamy of Bell's inequality explains the surprising behavior of nonviolation of the CHSH inequality by highly entangled particles in translation invariant systems. Together with the extension of the monogamy bound to the $N$-partite scenario, this  provides a full characterization of pairwise qubit  correlations in the presence of translation symmetry: bipartite entanglement 
is short-ranged, quantum discord is long-ranged, and nonlocal (Bell) correlations cannot be present whatsoever. Moreover, we have extended the tripartite monogamy bound to the $N$-partite scenario, showing
that the bound increases with N, allowing values close to the classical limit for distant particles,
but never violating it. We also obtained evidences that states saturating the bound become rare as
the number of particles increase. Higher-dimensional systems (qudits) remain as a further challenge. 

Note added - After finishing this work, Michael Wolf and Daniel Calvalcanti pointed out to us that using the results 
of \cite{Terhal03} we can strengthen our Theorem 1, since translation invariance implies that the reduced density matrix
of two spins has a (1,2)-symmetric quasi-extension (it is two-sharable). From the second theorem of Ref.~\cite{Terhal03}, 
such a state will not violate any Bell inequality with two measurements on one side and infinity measures on the other side.

\begin{acknowledgments}

We acknowledge the Brazilian agencies CNPq and FAPERJ for financial support. This 
work is also supported by the Brazilian National Institute for Science and 
Technology of Quantum Information (INCT-IQ).

\end{acknowledgments}

\vspace{1cm}
\section{Erratum}

There is a hidden hypothesis in Theorem 1. Besides requiring translation 
invariance as a main condition for the nonviolation of the CHSH inequality, 
it is also demanded to impose the following condition for the validity of Theorem 1: 
{\it for any two qubits labeled by  $1$ and $2$ separated by a distance $r$, there also exists 
a third independent qubit at distance r from qubit $1$}.
This condition trivially holds for either infinite chains or finite chains with an odd number 
of sites. However, it may be not obeyed for very specific qubit pairs in finite chains with an 
even number of sites. 

So, for the case of translation invariant finite chains with even number of sites, a violation 
of the CHSH inequality between spins $i$ and $i+N/2+1$ may occur. This is the farthest apart 
qubit pair. In particular, we expect such a violation to be highly unlikely in physical hamiltonian 
with short-range correlations.

All other results of the letter, including Theorem 2, remains valid. We should mention that this 
has actually been noticed by Ref. \cite{YuSun}.

\end{document}